\documentclass[a4paper,10pt]{article}
\usepackage[utf8x]{inputenc}

\usepackage{color}
\usepackage{a4wide}
\usepackage[fleqn]{amsmath}
\usepackage{amssymb}
\usepackage{amsthm}

\usepackage[shortlabels]{enumitem}
\setenumerate{itemsep=0.005ex}

\usepackage{tikz}
\usepackage{setspace}
\usetikzlibrary{arrows,automata,decorations.pathreplacing}
\theoremstyle{plain}
\newtheorem{theorem}{Theorem}[section]
\newtheorem{lemma}[theorem]{Lemma}
\newtheorem{corollary}[theorem]{Corollary}
\newtheorem{proposition}[theorem]{Proposition}
\theoremstyle{definition}
\newtheorem{definition}[theorem]{Definition}

\newcommand{\floor}[1]{\lfloor #1 \rfloor } 
\newcommand{\abs}[1]{\left| #1 \right|}
\newcommand{\measure}[0]{\mu}
\newcommand{\tuple}[1]{\langle #1 \rangle}
\newcommand{\limrun}[0]{\infty}
\newcommand{\even}[0]{\operatorname{even}}
\newcommand{\odd}[0]{\operatorname{odd}}
\newcommand{\trans}[1]{\mathchoice{\xrightarrow{#1}}{\xrightarrow{\smash{\lower1pt\hbox{$\scriptstyle #1$}}}}{\text{Error}}{\text{Error}}}

\newcommand{\alocc}[2]{|\!|#1|\!|_{#2}}
\newcommand{\occ}[2]{|#1|_{#2}}

\usepackage{soul,url}

\title{Finite-state independence}
\author{\begin{tabular}{lcr}
  Ver\'onica Becher & Olivier Carton& Pablo Ariel Heiber 
\end{tabular}}

\begin{document}

\maketitle

\begin{abstract}
  In this work we introduce a notion of independence based on finite-state
  automata: two infinite words are independent if no one helps to compress the
  other  using one-to-one finite-state transducers with auxiliary input.
  We prove that, as expected, the set of independent pairs of infinite
  words has Lebesgue measure $1$.  We show that the join of two independent
  normal words is normal.  However, the independence of two normal words is
  not guaranteed if we just require that their join is normal.  To prove
  this we construct a normal word $x_1x_2x_3\ldots$ where $x_{2n}=x_n$ for
  every~$n$. This construction has its own interest.
\end{abstract}

\section{Introduction}

In this work we introduce a notion of independence for pairs of infinite words, based 
on finite-state automata. We call it {\em finite-state independence}. 

The concept of independence appears in many mathematical theories,
formalizing that two elements are independent if they have no common parts.
In classical probability theory the notion of independence is defined for
random variables.  In the case of random variables with finite range the
notion of independence can be reformulated in terms of Shannon entropy
function: two random variables are independent if the entropy of the pair
is exactly the sum of the individual entropies.  An equivalent formulation
says that two random variables are independent if putting one as a
condition does not decrease the entropy of the other one.  In algorithmic
information theory the notion of independence can be defined for finite
objects using program-size (Kolmogorov--Chaitin) complexity. Namely, finite
words $x$ and~$y$ are independent if the program-size complexity of the
pair $(x,y)$ is close to the sum of program-size complexities of~$x$
and~$y$.  Equivalently, up to a small error term, two finite words are
independent if the program-size complexity of one does not decrease when we
allow the other as an oracle.  

Algorithmic information theory also defines
the notion of independence for random infinite words, as follows.  Recall
that, according to Martin-L\"of's definition, an infinite word is random if
it does not belong to any effectively null set; an equivalent
characterization establishes that an infinite word is random if its
prefixes have nearly maximal program-size complexity, which means that they
are incompressible with Turing machines.  Two random infinite words
$x_1x_2\ldots$ and $y_1y_2\ldots$ are independent if their join
$x_1y_1x_2y_2\ldots$ is random~\cite{vanLambalgen87,BST2016}, see
also  \cite[Theorem~3.4.6]{Nies08}  and \cite{Kautz1991} for  independence
 on stronger notions of randomness.  An equivalent definition establishes
that two random infinite words are independent if the program-size
complexity of the initial segment of one, conditioned on the other one, is
nearly maximal.  This means that one word remains incompressible even when
using the other one as an oracle.
See~\cite{Vitanyi2008,DowHir2010,SUV} for a thorough presentation of this
material.  While the notion of independence for random infinite words is
well understood, algorithmic information theory has not provided a fully
satisfactory definition of independence for arbitrary infinite words, see
the discussion in~\cite{CaludeZimand10}.

Here we scale down the notion of independence given by algorithmic
information theory by considering incompressibility by finite-state
automata instead of incompressibility by Turing machines.  Our
definition builds on the theory of finite-state compression ratio
introduced by Dai, Lathrop, Lutz and Mayordomo~\cite{Dai04}.  The
finite-state compression ratio of an infinite word indicates how much it
can be compressed by one-to-one finite-state transducers, which are
finite-state automata augmented with an output transition function such
that the automata input--output behaviour is one-to-one
(Huffman~\cite{Huffman52} called them lossless finite-state compressors).
The infinite words that can not be compressed 
are exactly the Borel normal words (this result was first known from
combining~\cite{Schnorr71} and~\cite{Dai04}, see~\cite{BecherHeiber13} for
a direct proof).  We say that two infinite words are finite-state
independent if one does not help to compress the other  using
finite-state transducers.  In Theorem~\ref{thm:measureI} we show that the
set of finite-state independent pairs of infinite words has Lebesgue
measure~$1$, giving an elementary proof.

As expected, the join of two finite-state independent normal words is
normal (Theorem~\ref{thm:join}).  However, independence of two normal words
is not guaranteed if we just require that their join is normal.  To show
this we construct a normal word $x$ that is equal to the word $\even(x)$
that consists of the symbols of $x$ at even positions
(Theorem~\ref{thm:n=2n}).  Thus, if $\odd(x)$ consists of the symbols of
$x$ at odd positions, both $\odd(x)$ and $\even(x)$ are normal, and their
join is normal. But $\odd(x)$ and $\even(x)$ are not independent: $\odd(x)$
equals $\odd(\even(x))$.
This phenomenon is not isolated:  Alexander Shen (personal communication, August 2016) 
proved that  for the  set of words $x$ such that    $x=\even(x)$,
a word is normal with  probability  $1$.

The notion of finite-state independence we present here is based just on
deterministic finite-state transducers.  It remains to investigate if
non-deterministic finite-state transducers operating with an oracle can
achieve different compression ratios.  In the case of finite-state
transducers with no oracle, it is already known that the deterministic and
the non-deterministic models compress exactly the same words, namely, the
non-normal words~\cite{BecherCartonHeiber15}.  Some other models also
compress exactly the same words, such as the finite-state transducers with
a single counter~\cite{BecherCartonHeiber15} and the two-way
transducers~\cite{CartonHeiber15}.  It is still unknown if there is a
deterministic push-down transducer that can compress some normal words.

It also remains the question of how  to characterize  finite-state independence
other than by the conditional  compression ratio in finite-state automata.
One would  like  a characterization in terms of  a complexity function based on finite automata
as those considered in  \cite{Shallit2001} and \cite{HKH2014}.

\section{Primary definitions}

Let $A$ be a finite set of symbols,   the alphabet.  
We write $A^\omega$ for the set of all infinite words over $A$
and $A^k$ stands for the set of all words of length~$k$. 
The length of a finite word $w$ is denoted by $|w|$.  The positions in finite and infinite words are numbered starting from~$1$.
To denote the symbol at position~$i$ of a word $w$ we write $w[i]$ and to denote the subword of 
$w$ from position~$i$ to~$j$ we write $w[i..j]$.
We  use the customary notation for asymptotic growth of functions saying that
 $f(n)$ is in $O(g(n))$  if  $\exists k>0 \ \exists n_0 \ \forall n>n_0 $, $ |f(n)| \leq k |g(n)|$.

\subsection{Normality}

A presentation of the definitions and basic results on normal sequences can be read form \cite{BC2017}.
Here we start by introducing the number of {\em occurrences} and the number of {\em aligned occurrences} of
a word~$u$ in a word~$w$. 
\begin{definition}
  For two words  $w$ and $u$, the number  of \emph{occurrences}
  of~$u$ in~$w$, denoted by $\occ{w}{u}$,  and the number  of \emph{aligned occurrences}
  of~$u$ in~$w$, denoted by $\alocc{w}{u}$, are defined as
  \begin{align*}
    \occ{w}{u}  & =|\{ i : w[i..i+|u|-1] = u \}|, \\ 
    \alocc{w}{u} & =|\{ i : w[i..i+|u|-1] = u \text{ and } i \equiv 1 \bmod |u|\}|.
  \end{align*}
\end{definition}
For example, $\occ{aaaaa}{aa}=4$ and $\alocc{aaaaa}{aa}=2$. 
Aligned occurrences are obtained 
by cutting $w$ in $|u|$-sized pieces starting from the left. Notice that the definition of aligned occurrences 
has the condition $i \equiv 1 \bmod |u|$ (and not $i \equiv 0 \bmod |u|$), 
because the positions are numbered starting from~$1$.  
Of course, when a word $u$ is just a symbol, $\occ{w}{u}$ and $\alocc{w}{u}$ coincide. 
Aligned occurrences can be seen as  symbol occurrences using a power alphabet: 
if~$w$ is a word whose  length is a  multiple of~$r$, then $w$ can be considered as a word $\pi(w)$ over~$A^r$ by grouping 
its symbols into blocks of length~$r$. 
The aligned occurrences of a word~$u$ of length~$r$ in 
$w$ then correspond to the occurrences of the symbol $\pi(u)$ 
in the word~$\pi(w)$, and $\alocc{w}{u} = \occ{\pi(w)}{\pi(u)}$.

We recall the definition of Borel normality~\cite{Borel1909} for infinite words 
(see the books~\cite{Bugeaud12,KuiNie74} for a complete presentation). 
An infinite word $x$ is \emph{simply normal} to word length~$\ell$ if all the blocks of length $\ell$ 
have asymptotically the same frequency of aligned occurrences, i.e., if for every $u \in A^\ell$,
\begin{displaymath}
  \lim_{n \to \infty} \frac{\alocc{x[1..n]}{u}}{n/\ell} = |A|^{-\ell}.
\end{displaymath}
An infinite word $x$ is \emph{normal} if it is simply normal to every length.  
Normality is defined here in terms of aligned occurrences but it could also 
be defined in terms of all occurrences.  The equivalence between the two
definitions requires a proof (see Theorem 4.5 in~\cite{Bugeaud12}).

\subsection{Automata}

We  consider \emph{$k$-tape automata},  also known as
\emph{$k$-tape transducers} when $k$ is greater than~$1$ \cite{PerrinPin04,Sakarovitch09}. 
We  call them  $k$-automata and we consider them  for $k$ equal to $1$, $2$, or $3$.  
A $k$-automaton is a tuple
${\mathcal T }= \tuple{Q,A,\delta,I}$, where $Q$ is the finite state set, $A$ is the
alphabet, $\delta$ is the transition relation, $I$ is the set of initial
states.  The transition relation is a subset of
$Q \times (A \cup \{\varepsilon\})^k \times Q$.  A transition is thus a
tuple $\tuple{p,\alpha_1,\ldots,\alpha_k,q}$ where $p$ is its starting state,
$\tuple{\alpha_1,\ldots,\alpha_k}$ is its \emph{label} and $q$ is its ending
  state.  Note that each $\alpha_i$ is here either a symbol of the alphabet or
the empty word.  A~transition is written as $p \trans{\alpha_1,\ldots,\alpha_k} q$.  As
usual, two transitions are called \emph{consecutive} if the ending state of the
first is the starting state of the second.

\begin{figure}[htbp]
  \begin{center}
    \begin{tikzpicture}[scale=1]
    % États
    \node (state) at (-1,2.5) [shape=rectangle,draw,rounded corners=1mm,
                            inner sep=10] {$Q$};
    \begin{scope}
      % Lignes horizontales
      \draw (0.5,3.5)  -- (4.2,3.5);
      \draw[dotted] (4.2,3.5) -- (4.5,3.5);
      \draw (0.5,4)    -- (4.3,4);
      \draw[dotted] (4.3,4) -- (4.5,4);
      % Lignes verticales
      \foreach \x in {0.5,1,1.5,2,2.5,3,3.5,4}
        \draw (\x,3.5) -- (\x,4);
      % Tête de lecture
      \draw[very thick] (2,3.5) rectangle (2.5,4);
      \draw[->,>=latex,rounded corners] (state) -| ++(1,0.75) -| (2.25,3.5);
      % Mot d'entrée
      \foreach \x/\xtext in {0.75/1,1.25/2,1.75/3,2.25/4,2.75/5,3.25/6,3.75/7}
        \node at (\x,3.7) {$a_{\xtext}$};
    \end{scope}
    \begin{scope}[shift={(0,-1)}]
      % Lignes horizontales
      \draw (0.5,3.5)  -- (4.2,3.5);
      \draw[dotted] (4.2,3.5) -- (4.5,3.5);
      \draw (0.5,4)    -- (4.3,4);
      \draw[dotted] (4.3,4) -- (4.5,4);
      % Lignes verticales
      \foreach \x in {0.5,1,1.5,2,2.5,3,3.5,4}
        \draw (\x,3.5) -- (\x,4);
      % Tête de lecture
      \draw[very thick] (1,3.5) rectangle (1.5,4);
      \draw[->,>=latex,rounded corners] (state) -| ++(1,-0.25) -| (1.25,3.5);
      % Mot d'entrée
      \foreach \x/\xtext in {0.75/1,1.25/2,1.75/3,2.25/4,2.75/5,3.25/6,3.75/7}
        \node at (\x,3.7) {$b_{\xtext}$};
    \end{scope}
    \begin{scope}[shift={(0.5,1)}]
      %% Bande de sortie
      \draw (0,0.5)  -- (4.3,0.5);
      \draw[dotted] (4.3,0.5) -- (4.5,0.5);
      \draw (0,0)    -- (4.2,0);
      \draw[dotted] (4.2,0) -- (4.5,0);
      % Lignes verticales
      \foreach \x in {0,0.5,1,1.5,2,2.5,3,3.5,4}
        \draw (\x,0) -- (\x,0.5);
      % Tête de lecture
      \draw[very thick] (3,0) rectangle (3.5,0.5);
      \draw[->,>=latex,rounded corners] (state) -| ++(1,-0.75) -| (3.25,0.5);
      %\draw[->,thick] (3.5,0.75) -- (3.8,0.75);
      % Mot d'entrée
      \foreach \x/\xtext in {0.25/1,0.75/2,1.25/3,1.75/4,2.25/5,2.75/6,3.25/7,3.75/8}
        \node at (\x,0.2) {$c_{\xtext}$};
    \end{scope}
    \end{tikzpicture}
  \end{center}
  \caption{A $3$-automaton and its tapes}.
  \label{fig:principle}
\end{figure}
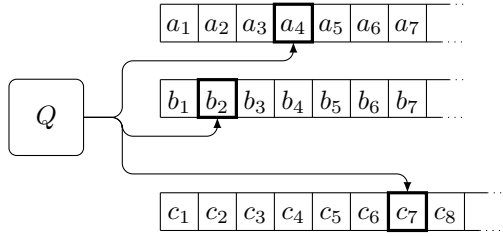

An infinite run  is an infinite sequence of consecutive transitions
\begin{displaymath}
  q_0 \trans{\alpha_{1,1},\ldots,\alpha_{k,1}} q_1 
      \trans{\alpha_{1,2},\ldots,\alpha_{k,2}} q_2 
      \trans{\alpha_{1,3},\ldots,\alpha_{k,3}} q_3 \trans{}\cdots 
\end{displaymath}
The label of the run is the component-wise concatenation of the
labels of the transitions given by  the tuple
$\tuple{x_1,\ldots,x_k}$ where each $x_j$ for $1 \le j \le k$ is equal to
$\alpha_{j,1}\alpha_{j,2}\alpha_{j,3}\cdots$.  Note that some label $x_j$ might be finite
although the run is infinite since some transitions may have empty
labels. The run is accepting if its first state~$q_0$ is initial and each
word~$x_j$ is infinite.  Such an accepting run is written shortly as
$q_0 \trans{x_1,\ldots,x_k}  \limrun$.
The tuple $\tuple{x_1,\ldots,x_k}$ is accepted if there
exists at least one accepting run with label $\tuple{x_1,\ldots,x_k}$.
The $1$-automata are the usual automata with $\varepsilon$-transitions
and the $2$-automata are the usual automata with input and output also
known as transducers.

In this work we consider only deterministic $k$-automata.  We actually
consider $k$-automata where the transition is determined by a subset of the
$k$ tapes.  Informally, a $k$-automaton is $\ell$-deterministic, for
$1\leq \ell\leq k$, if the run is entirely determined by the contents of
the first $\ell$ tapes.  More precisely, a $k$-automaton is
$\ell$-deterministic if the following two conditions are fulfilled,
\begin{itemize} \itemsep0cm
\item[-] the set $I$ of initial states is a singleton set;
\item[-] for any two transitions $p \trans{\alpha_1,\ldots,\alpha_k} q$ and 
  $p' \trans{\alpha'_1,\ldots,\alpha'_k} q'$ with $p = p'$, 

\qquad  if $\alpha_j = \varepsilon$ for some $1 \le j \le \ell$, then $\alpha'_j =
    \varepsilon$

\qquad if $\alpha_1 = \alpha'_1, \ldots, \alpha_{\ell} = \alpha'_{\ell}$, then 
    $\alpha_{\ell+1} = \alpha'_{\ell+1}, \ldots,\alpha_n = \alpha'_n$ and $q = q'$.
\end{itemize}
The conditions on the transitions leaving a state~$p$ are the following.
The first one requires that, among the first $\ell$ components, the ones
with empty label are the same for all transitions leaving~$p$.  This means
that each state determines (among the first $\ell$ tapes) the tapes from
which a symbol is read (the ones with a symbol as label) and the tapes from
which no symbol is read (the ones with empty label).  The second condition
is the usual one stating that two transitions leaving~$p$ and with the same
labels in the first $\ell$ components must be the same.

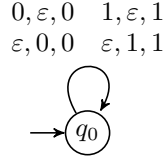
\begin{figure}[htbp]
  \begin{center}
    \begin{tikzpicture}[->,>=stealth',initial text=,semithick,auto,inner sep=2pt]
      \tikzstyle{every state}=[minimum size=0.4]
      \node[state,initial left] (q0) at (0,0)  {$q_0$};
      \path (q0) edge[out=120,in=60,loop] node {$\begin{array}{c} 
                                       0,\varepsilon,0 \quad 1,\varepsilon,1 \\
                                       \varepsilon,0,0 \quad \varepsilon,1,1
                                                \end{array}$} ();
    \end{tikzpicture}
  \end{center}
  \caption{A non-deterministic $3$-automaton for the shuffle.}
  \label{fig:shuffle}
\end{figure}

\begin{figure}[htbp]
  \begin{center}
    \begin{tikzpicture}[->,>=stealth',initial text=,semithick,auto,inner sep=2pt]
      \tikzstyle{every state}=[minimum size=0.4]
      \node[state,initial above] (q0) at (0,0)  {$q_0$};
      \node[state]  (q1) at (3,0) {$q_1$};
      \path (q0) edge[bend left=20] node {$\begin{array}{c} 
                                             0,\varepsilon,0 \\
                                             1,\varepsilon,1
                                          \end{array}$} (q1);
      \path (q1) edge[bend left=20] node {$\begin{array}{c} 
                                             \varepsilon,0,0 \\
                                             \varepsilon,1,1
                                          \end{array}$} (q0);
    \end{tikzpicture}
  \end{center}
  \caption{A $2$-deterministic $3$-automaton for the join.}
  \label{fig:join}
\end{figure}
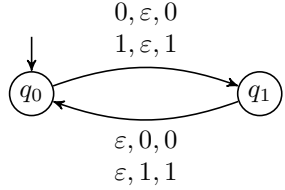

Figures~\ref{fig:shuffle} and~\ref{fig:join}
show  $3$-automata that accept a triple
$\tuple{x,y,z}$ of infinite words over the alphabet $\{ 0, 1 \}$.
In   Figure~\ref{fig:shuffle} the tuple
 $\tuple{x,y,z}$ is accepted if  $z$ is a shuffle  of $x$ and~$y$.
In  Figure~\ref{fig:join} the tuple
 $\tuple{x,y,z}$ is accepted 
if  $z$ is the join of $x$ and~$y$  (the join of two infinite words
$x = a_1a_2a_3\cdots$ and $y = b_1b_2b_3\cdots$ is the infinite word
$z = a_1b_1a_2b_2a_3\cdots$.)
The $3$-automaton in Figure~\ref{fig:join} is $2$-deterministic
but the one in Figure~\ref{fig:shuffle} is not
because the two transitions $q_0 \trans{0,\varepsilon,0} q_0$ and
$q_0 \trans{\varepsilon,0,0} q_0$ violate the required condition.

We assume each $\ell$-deterministic $k$-automata ${\mathcal T}$
computes a partial  function $(A^\omega)^\ell\to (A^\omega)^{k-\ell}$.  
Let $\mathcal{T}$ be  an $\ell$-deterministic  $k$-automaton.
For each tuple $\tuple{x_1,\ldots,x_{\ell}}$ of infinite words there
exists at most one tuple $\tuple{y_{\ell+1},\ldots,y_k}$ of infinite words such
that the $k$-tuple \linebreak
$\tuple{x_1,\ldots,x_{\ell},y_{\ell+1},\ldots,y_k}$ is
accepted by~$\mathcal{T}$.  The automaton~$\mathcal{T}$ realizes then a 
partial function from $(A^\omega)^{\ell}$ to $(A^\omega)^{k-\ell}$ and the tuple
$\tuple{y_{\ell+1},\ldots,y_k}$ is denoted $\mathcal{T}(x_1,\ldots,x_{\ell})$.
The $1$-deterministic $2$-automata are also called \emph{sequential transducers}
in the literature.  When a $k$-automaton is $\ell$-deterministic, 
we write the transition as  $p \trans{\alpha_1,\ldots,\alpha_{\ell}|\beta_{\ell+1},\ldots,
  \beta_k} q$ to emphasize the fact that the first $\ell$ tapes are input
tapes and that the $k-\ell$ remaining ones are output tapes.

We comment here on transitions reading no symbol from the input tapes.  Let
$\mathcal{A}$ be a $\ell$-deterministic $k$-automaton.  Suppose that there
exists a transition from state~$p$ to state~$q$ whose label has the form
$\tuple{\varepsilon,\ldots,\varepsilon,\beta_{\ell+1},\ldots, \beta_k}$
with empty labels in the first $\ell$ components.  Let us call these such a
transition an $\varepsilon^\ell$-transition.  We claim that it is always
possible to get rid of~$\varepsilon^\ell$-transitions. The
automaton~$\mathcal{A}$ being deterministic implies that this transition is
the only transition leaving state~$p$.  If there exists a cycle made
of~$\varepsilon^\ell$-transitions, this cycle is a dead end in the
automaton and its states can be removed without changing the set of
accepting runs of~$\mathcal{A}$.  Let us recall that it is required that
all labels of an accepting state to be be infinite.  Assume now that there
is no cycle of such $\varepsilon^\ell$-transitions.  Removing the
transition
$p \trans{\varepsilon,\ldots,\varepsilon|\beta_{\ell+1},\ldots, \beta_k} q$
and adding a transition
$p \trans{\alpha_1,\ldots,\alpha_\ell|\beta_{\ell+1}\gamma_{\ell+1},\ldots,
  \beta_k\gamma_k} r$ for each transition
$q \trans{\alpha_1,\ldots,\alpha_\ell|\gamma_{\ell+1},\ldots, \gamma_k} r$
leaving~$q$ preserve the accepting paths and decrease the number of
$\varepsilon^\ell$-transitions.  Completing the process until no
$\varepsilon^\ell$-transition remains remove all of them.  In the rest of
the paper, we always assume that $\varepsilon^\ell$-transitions have been
removed.

Let $\mathcal{T}$ be a $1$-deterministic $2$-automaton.  To define the compression
ratio of an infinite word~$x=a_1 a_2 \ldots $ by $\mathcal T$, denoted by $\rho_{\mathcal{T}}(x) $,
    consider the unique accepting run 
\[
q_0 \trans{a_1|v_1} q_1 \trans{a_2|v_2} q_2 \trans{a_3|v_3} q_3 \cdots
\]
 of~$\mathcal{T}$,
where each $a_i$ is a symbol in $A$, and each $v_i$ is a finite word (possibly empty) of symbols in $A$.
Then, 
\begin{displaymath}
   \rho_{\mathcal{T}}(x) = 
         \liminf_{n \to \infty} \frac{|v_1v_2\cdots v_n|}{|a_1a_2\cdots a_n|}.
\end{displaymath} 
For a given infinite word  $x$ it may happen that for   some automata $\mathcal{T}$, $   \rho_{\mathcal{T}}(x)$ is greater than~$1$ and
for some other  automaton $\mathcal{T}'$,   $\rho_{\mathcal{T}'}(x)$ is less than $1$.
We say that an infinite word $x$ is \emph{compressible } by a $1$-deterministic $2$-automaton~$\mathcal{T}$ if
$\rho_{\mathcal{T}}(x) < 1$.  
A $1$-deterministic $2$-automaton~$\mathcal{T}$ is called \emph{one-to-one} if the  function which maps
$x$ to~$\mathcal{T}(x)$ is one-to-one. 
The \emph{compression ratio} of an infinite word~$x$ is the infimum of the compression ratios achievable by all
one-to-one $1$-deterministic $2$-automata:
\begin{displaymath}
  \rho(x) = \inf\{\rho_{\mathcal{T}}(x) :  
   \text{$\mathcal{T}$ is a one-to-one $1$-deterministic $2$-automaton}\}
\end{displaymath}
This compression ratio~$\rho(x)$ is always less or equal to~$1$, as
witnessed by the one-to-one compressor~$\mathcal{C}_0$ which copies each
symbol of the input to the output.  For each infinite word~$x$, the
compression ratio~$\rho_{\mathcal{C}_0}(x)$ is equal to~$1$.

The sequence $x = 0^{\omega}$, made entirely of zeros, has compression
ratio $\rho(x) = 0$.  This is because for each positive real
number~$\varepsilon$, there exists a compressor~$\mathcal{C}$ such that
$\rho_{\mathcal{C}}(x) < \varepsilon$. However, 
the compression ratio~$0$ is not achievable by any one-to-one $1$-deterministic
$2$-automaton $\mathcal{C}$ for the following reason.
Every such automaton $\mathcal{C}$ computes a function  $A^\omega\to A^\omega$,
and the compression ratio by~$\mathcal{C}$ is the ratio between the output
and the input in the cycle reached by the infinite run.  
If the input word $x$ is ultimately periodic as for $x = 0^\omega$ 
and $x$ is in the domain of ${\cal C}$
then the run on~$x$ is also ultimately periodic, 
and the compression ratio of $x$ by~$\mathcal{C}$ is non-zero.  
On the other extreme, the words with compression ratio equal to~$1$ are exactly the normal words.

\section{Finite-state independence}

To define the notion of finite-state independence we  introduce 
 the  notion of conditional compression  ratio.
The  conditional compression  ratio of an
infinite word~$x$ with respect to another infinite word~$y$ is the ratio of
compression of~$x$ when $y$ is used as an oracle. 
To define this notion, we consider
 $2$-deterministic $3$-automata such that two input tapes
contain the words $x$ and~$y$ and the output tape contains the result of
the compression of~$x$ with the help of~$y$.

  A \emph{compressor} is a $2$-deterministic $3$-automata~$\mathcal{C}$
  such that for any fixed infinite word~$y$, the  function
  $x \mapsto \mathcal{C}(x,y)$ which maps $x$ to the
  output~$\mathcal{C}(x,y)$ is one-to-one.
This guarantees that, if $y$ is known, $x$ can be recovered from
$\mathcal{C}(x,y)$.  Note that we do \emph{not} require 
that the function $(x,y) \mapsto \mathcal{C}(x,y)$ is one-to-one, 
which  would be a much stronger requirement.

\begin{definition}\label{def:rho}
Let $\mathcal{C}$ be a compressor.  The \emph{conditional compression
  ratio} of an infinite word~$x$ with respect to~$y$ for~$\mathcal{C}$ is
determined by the unique accepting run
\begin{displaymath}
  q_0 \trans{\alpha_1,\beta_1|w_1} q_1 
      \trans{\alpha_2,\beta_2|w_2} q_2 
      \trans{\alpha_3,\beta_3|w_3} q_3 \cdots
\end{displaymath}
such that $x = \alpha_1\alpha_2\alpha_3\cdots$ and $y = \beta_1 \beta_2 \beta_3\ldots$, 
with each $\alpha_i, \beta_i\in (A\cup\epsilon)$, as
\begin{displaymath}
   \rho_{\mathcal{C}}(x/y) = 
         \liminf_{n \to \infty} \frac{|w_1w_2\cdots w_n|}{|\alpha_1 \alpha_2\cdots \alpha_n|}.
\end{displaymath} 
Notice that the number of symbols read from~$y$, the length of $\beta_1 \beta_2\cdots \beta_n$,
is \emph{not} taken into account when defining~$\rho_{\mathcal{C}}(x/y)$.

The \emph{conditional compression ratio} of an infinite word $x$ given an infinite word~$y$,
denoted by $\rho(x/y)$, is the infimum of the compression ratios
$\rho_{\mathcal{C}}(x/y)$ of all compressors~$\mathcal{C}$ with input~$x$ and oracle~$y$.
\end{definition}

Notice that the plain compression ratio $\rho(x)$ and the conditional compression ratio $\rho(x/y)$ 
always exist and they are values between $0$ and $1$ (witnessed by the identity function).

The following proposition gives sufficient conditions for the maximum compression, 
when the conditional compression ratio is equal to  zero.

\begin{proposition} \label{pro:function}
If a function~$f$
is realizable by a $1$-deterministic $2$-automata then, for every~$x$,
  $\rho(f(x)/x)=0$.
\end{proposition}
\begin{proof}
Assume that the function~$f$ is realized by a $1$-deterministic
  $2$-automaton~$\mathcal{T}$, so
$f(x) = \mathcal{T}(x)$ for every infinite word $x$.
 We fix a positive integer~$k$ and we construct a
  compressor~$\mathcal{C}$ such that $\rho_{\mathcal C}(f(x)/x)=1/k$.
This compressor has two input tapes, the
  first one containing a word~$y$ and the second one the word~$x$.  It
  compresses $y$ in a non-trivial way only when $y$ is equal to~$f(x)$.
  The automaton~$\mathcal{C}$ proceeds as follows.  It reads the infinite
  word~$x$ from the second input tape and computes $f(x)$ by simulating the
  automaton~$\mathcal{T}$.  While $f(x)$ coincides with $y$ for the next
  $k$ symbols, then $\mathcal{C}$ writes a $0$.  When there is a mismatch,
  $\mathcal{C}$ writes a $1$ and then copies the remaining part of~$y$ to
  the output tape.  Thus, if the mismatch occurs at position
  $m = kp + r$ with $1 \leq r \leq k$,  the automaton $\mathcal{C}$ writes $p$ symbols~$0$
  before the mismatch, a symbol~$1$ and $y_{kp+1}y_{kp+2}y_{kp+3} \cdots$.
\end{proof}

The following proposition provides a sufficient condition for compressing
$x$ given~$y$: some correlation between symbols in $x$
and~$y$ at the same positions ensures $\rho(x/y) < 1$.

\begin{proposition} \label{pro:correl}
  Let $x$ and $y$ be two infinite words.  If there are three symbols $c$,
  $c'$ and~$d$ and an increasing sequence $(m_n)_{n\ge0}$ of integers such
  that
  \begin{displaymath}
    \lim_{n\to\infty}{\frac{|\{ 1 \le i \le m_n : x[i] = c, y[i] = d\}|}{m_n}} 
    \neq
    \lim_{n\to\infty}{\frac{|\{ 1 \le i \le m_n : x[i] = c', y[i] = d\}|}{m_n}},
  \end{displaymath}
  then $\rho(x/y) < 1$.
\end{proposition}
We assume here that both limits exist; however,
the same result holds if just one or none of the two limits  exist.
\begin{proof}
  By replacing the sequence $(m_n)_{n\ge0}$ by some subsequence of it, 
  we may assume without loss of generality that 
  $\lim_{n\to\infty}{|\{ 1 \le i \le m_n : x[i] = a, y[i] = b\}|/m_n}$
  exists for arbitrary symbols $a$ and~$b$.  
  This limit is denoted by $\pi(a,b)$. The existence of all the limits $\pi(a,b)$ implies that 
   for each symbol~$b$,  
\[
\lim_{n\to\infty}{|\{ 1 \le i \le m_n : y[i] = b \}|/m_n}
\] also
  exists and 
  \[
 \lim_{n\to\infty}{\frac{|\{ 1 \le i \le m_n : y[i] = b \}|}{m_n}} =
    \sum_{a \in A} \pi(a,b).
\]
  This limit is denoted by $\pi_y(b)$.
  Define 
\[
\nu(a/b)=\left\{
\begin{array}{ll}
 \pi(a,b)/\pi_y(b)&\text{ if $\pi_y(b) \neq 0$ }
\\ 
1/|A|& otherwise.  
\end{array}
\right.
\]
  Let $k$ be a block length to be fixed later.  For two words
  $u = a_1 \cdots a_k$ and $v = b_1 \cdots b_k$ of length~$k$, define
\[
\nu(u/v) = \prod_{i=1}^k\nu(a_i/b_i).
\]
Let us recall that a set $P$ of words is prefix-free if no distinct words
$u,v \in P$ satisfy that $u$ is a prefix of $v$.  Note that if $P$ is
prefix-free, every word~$w \in A^*$ has at most one factorization
$w = u_1 \cdots u_n$ where each $u_i$ belongs to~$P$.  We will use the
following well-known fact due to Huffman \cite{Huffman52}.  

Let
$p_1,\ldots,p_n$ be real numbers such that $0 \le p_i \le 1$ for each
$1 \le i \le n$ and $\sum_{i=1}^{n}{p_i} = 1$, then there exist $n$
distinct words $u_1,\ldots,u_n$ such that the set $\{ u_1, \ldots, u_n \}$
is prefix-free and $|u_i| \le \lceil -\log p_i \rceil$ for $1 \le i \le n$.

It is purely routine to check that $\sum_{u \in A^k} \nu(u/v) = 1$ for each
word~$v$ of length~$k$.  Since $\sum_{u \in A^k} \nu(u/v) = 1$, there
exists for each word~$v$, a prefix-free set $\{ w(u,v) : u,v \in A^k\}$
such that the relation $|w(u,v)| \le \lceil -\log_{|A|} \nu(u/v)\rceil$
holds for each $u$ and~$v$.  These words $w(u,v)$ are used by the
transducer to encode the infinite word~$x$ with the help of~$y$.  The
transducer reads $x$ and~$y$ by blocks of length~$k$.  For each pair of
blocks $u$ and~$v$, it outputs $w(u,v)$.  This output can be decoded with
the help of~$y$ because for each fixed block~$v$, the possible blocks~$u$
are in one-to-one correspondence with the words~$w(u,v)$.

We now evaluate the length of the output of the transducer.  Let $p(n,u,v)$
be the number of occurrences of the pair $(u,v)$ in $x$ and $y$.  Then,
  \begin{displaymath}
    p(n,u,v) = |\{ 1 \le i \le n-k: i = 1\bmod k, x[i..i+k-1] = u, y[i..i+k-1] = v\}|
  \end{displaymath}
  \begin{align*}
    \rho(x/y) 
     & \le \lim_{n\to\infty} \frac{1}{n}       
           \sum_{u,v \in A^k} p(n,u,v)|w(u,v)| \\
     & \le \lim_{n\to\infty} \frac{1}{n} 
           \sum_{u,v \in A^k} p(n,u,v) \lceil -\log_{|A|} \nu(u/v)\rceil \\
     & \le \frac{1}{k} + \lim_{n\to\infty} 
            \frac{1}{n}\sum_{u,v \in A^k} p(n,u,v)(-\log_{|A|} \nu(u/v)) \\
     & = \frac{1}{k} + \lim_{n\to\infty} 
            \frac{1}{n}\sum_{\substack{
                              u = a_1 \cdots a_k \\
                              v = b_1 \cdots b_k}}
            p(n,u,v)(-\log_{|A|} \prod_{i=1}^k\nu(a_i/b_i))   \\ 
     & = \frac{1}{k} + \lim_{n\to\infty} 
            \frac{1}{n}\sum_{a,b\in A} |\{ 1 \le i \le n : x[i] = a, y[i]=b\}|
            \log_{|A|} \frac{1}{\nu(a/b)} \\
     & =  \frac{1}{k} + 
            \sum_{b\in A, \pi_y(b)\not=0} \pi_y(b) \sum_{a\in A} \frac{\pi(a,b)}{\pi_y(b)}
            \log_{|A|} \frac{\pi_y(b)}{\pi(a,b)} 
  \end{align*}
The previous to the last row results from 
counting correlated symbols instead of counting correlated blocks of length $k$.
The last  equality results  from using $\nu(a/b) =   \frac{\pi(a,b)}{\pi_y(b)}$ and    
applying the definition of $\pi(\cdot,\cdot)$ and $\pi_y(\cdot)$.
  
We have that for each symbol~$b$, the sum
\[  \sum_{a\in A} \frac{\pi(a,b)}{\pi_y(b)}\log_{|A|}
  \frac{\pi_y(b)}{\pi(a,b)} \leq 1.
\]
However, for $b = d$, the above sum  is strictly less than~$1$.  
Then, it  follows  that
  \begin{displaymath}
    \sum_{b\in A} \pi_y(b) \sum_{a\in A} \frac{\pi(a,b)}{\pi_y(b)}
            \log \frac{\pi_y(b)}{\pi(a,b)} < 1.
  \end{displaymath}
  If $k$ is chosen great enough, the conditional compression
  ratio~$\rho(x/y)$ satisfies $\rho(x/y) < 1$.  
\end{proof}

The definition of finite-state independence of two infinite words is based on the conditional compression
ratio.

\begin{definition}
  Two infinite words $x$ and~$y$ are \emph{finite-state independent} if
  $\rho(x/y) = \rho(x)$, \linebreak $\rho(y/x) = \rho(y)$ and the compression ratios
  of $x$ and~$y$ are non-zero.
\end{definition}  

Note that we require that the compression ratios of $x$ and~$y$ are
non-zero.  This means that a word~$x$ such that $\rho(x) = 0$ 
cannot be part of  an independent pair.
Without this requirement, every two words $x$ and~$y$
such that $\rho(x) = \rho(y) = 0$ would be independent.  In particular,
every word~$x$ with $\rho(x) = 0$ would be independent of itself.

\section{Join independence is not enough}

Recall that the \emph{join} of two infinite words
$x = a_1a_2a_3\cdots$ and $y = b_1b_2b_3 \cdots$ is the infinite word
$a_1b_1a_2b_2\cdots$ obtained by interleaving their symbols.  It is denoted
$x \vee y$.  A possible definition of independence for normal words $x$
and~$y$ would be to require their join $x \vee y$ to be normal.  We call
this notion \emph{join independence}.  This definition of independence
would be natural since it mimics the definition of independence in the
theory of algorithmic randomness, where two random infinite words are
independent if their join is random~\cite{vanLambalgen87}, see also
\cite[Theorem~3.4.6]{Nies08}.
One can ask whether a similar result holds for our definition.
Is it true that  two normal words are finite-state independent
 if and only if their join is normal?
It turns out that 
finite-state independence implies join independence. 
But the converse fails.

We use the following notation.  For a given infinite word
$x = a_1a_2a_3\ldots$, we define infinite words $\even(x)$ and $\odd(x)$ 
as  $a_2a_4a_6\cdots$ and $a_1a_3a_5\cdots$ respectively.  Similarly, for a
finite word~$w$, $\even(w)$ and $\odd(w)$ are the words
appearing on the even and the odd positions in~$w$. For example,
$x=\even(x)$ means that $a_n=a_{2n}$ for all $n$.

\begin{theorem} \label{thm:join}
  Let $x$ and $y$ be normal. 
If $x$ and $y$ are finite-state independent then $x \vee y$  is normal.
\end{theorem}

\begin{proof}
Suppose $x\vee y$ is not normal. 
Then, there is a length~$k$ such that a word of length~$k$ 
does not have aligned frequency $2^{-k}$  in~$x \vee y$.  
We can assume without loss of generality that the length 
is even (increasing it if necessary), so we assume that
 some word of length~$2k$ does not have aligned frequency~$2^{-2k}$.  
Split $x\vee y$ into blocks of size $2k$:
\[
x \vee y = (u_1\vee v_1)(u_2\vee v_2)\cdots
\]
where $x=u_1u_2\ldots$ and $y=v_1v_2\ldots$ are the respective 
decomposition in $k$-blocks. 
Consider a sequence of positions $(m_n)_{n\geq 1}$ in $x\vee y$ 
such that each block of length $2k$ has an aligned frequency,
and let   $u\vee v$ be a   block of length~$2k$ 
whose frequency  $f$ along $(m_n)_{n\geq 1}$ is not $2^{-2k}$.
Let  $u'\vee v$   be another  block of length~$2k$
with  frequency along $(m_n)_{n\geq 1}$ 
different from $f$.
(Notice that if  $u\vee v$ and all blocks $u'\vee v$ have the all 
the same  frequency $f$ along $(m_n)_{n\geq 1}$  then, necessarily,  
there is a block $v'$ such that  $u\vee v'$  has frequency 
along $(m_n)_{n\geq 1}$ different from $f$, so we  exchange the  roles of~$x$ and $y$).
Then,  at the positions $(m_n/2)_{n\geq 1}$
 \begin{align*}
 &   \lim_{n\to\infty}{\frac{|\{ 1 \le i \le m_n/2 : x[i..i+k-1] = u, y[i..+k-1] = v\}|}{m_n}} 
    \neq
\\
 &   \lim_{n\to\infty}{\frac{|\{ 1 \le i \le m_n/2 : x[i..+k-1] = u', y[i..i+ k -1] = v\}|}{m_n}},
  \end{align*}
By the argument in Proposition~\ref{pro:correl} 
we conclude that   $\rho(u_1u_2\ldots/v_1v_2\ldots)<1$. 
By considering a compressor that reads  blocks of length $k$ we obtain
\[
\rho(x/y)=\rho(u_1u_2\ldots/v_1v_2\ldots)<1.
\]
\end{proof}

Actually, the proof of Theorem~\ref{thm:join} gives  a stronger result.

\begin{proposition}
If $y$ is normal and $\rho(x/y)=1$, then $x\vee y$ is normal.
\end{proposition}

We now show that   there are two words, that are join independent 
but not finite-state independent.  
The proof is based on the existence of 
normal word $x$ such that  $x = \even(x)$,
that we prove in Theorem~\ref{thm:n=2n} below.

\begin{theorem}  
There exists two normal words $x$ and $y$ such that $x \vee y$ is normal
  but $x$ and $y$ are not independent.
\end{theorem}

\begin{proof}
 By Theorem~\ref{thm:n=2n}, proved below,  there is a 
normal word $x$ such that  $x = \even(x)$.
Let  $y = \odd(x)$ and $z = \even(x)$. 
 Since $x$ is normal and $x = y \vee z$, the words $y$ and~$z$ are join independent.  
Since $y=\odd(x)$ and $x= \even(x)= z$,  we  have the equality    $y = \odd(z)$. 
This implies, by Proposition~\ref{pro:function}, that
$\rho(y/z) = 0$; hence,   $y$ and~$z$ are not finite-state independent.
%  Let $y$ a normal word be such that $y = \even(y)$ as given by the
%  previous theorem and set $x = \odd(y)$.  By definition
%  $x \vee y = \odd(y) \vee \even(y) = y$ is normal.  Since $x = \odd(y)$
%  it is easy to compress $x$ using $y$.  Therefore $x$ and $y$ are not
%  independent.
\end{proof}

\subsection{Construction of a normal word $x$  such that $x = \even(x)$}

\begin{theorem} \label{thm:n=2n}
There is a normal word $x$ such that  $x = \even(x)$.
\end{theorem}

%The previous theorem has the following corollary which shows that
%independence is indeed stronger than requiring the join to be normal.
%\begin{corollary}
%  There exists two normal words $x$ and $y$ such that $x \vee y$ is normal
%  but $x$ and $y$ are not independent.
%\end{corollary}
%\begin{proof}
%  Let $y$ a normal word be such that $y = \even(y)$ as given by the
%  previous theorem and set $x = \odd(y)$.  By definition
%  $x \vee y = \odd(y) \vee \even(y) = y$ is normal.  Since $x = \odd(y)$
%  it is easy to compress $x$ using $y$.  Therefore $x$ and $y$ are not
%  independent.
%\end{proof}

Here we prove this existence by giving an explicit construction
of  a normal word~$x = a_1a_2a_3\cdots$
 over the alphabet~$\{ 0, 1\}$ 
such that $a_{2n} = a_n$ for every~$n$.
The construction can be easily extended to an alphabet of size~$k$ to
obtain a word $a_1a_2a_3\cdots$ such that $a_{kn} = a_n$ for each integer
$n \ge 1$.  Beware that the construction, as it is presented below, cannot
provide a word $a_1a_2a_3\cdots$ over a binary alphabet such that
$a_{3n} = a_n$ 
(some more sophisticate one is needed, but it can be done; the probabilistic argument also works).

A finite word~$w$ is called \emph{$\ell$-perfect} for an
integer~$\ell \ge 1$, if $|w|$ is a multiple of~$\ell$ and 
all words of length $\ell$ have the same number of aligned occurrences 
$|w|/(\ell2^\ell)$ in $w$.% ---because there are exactly $2^\ell$ words of length~$\ell$.
  
\begin{lemma} \label{lem:doublelength}
  Let $w$ be an $\ell$-perfect word such that $|w|$ is a multiple of
  $\ell2^{2\ell}$.  Then, there exists a $2\ell$-perfect word~$z$ of length
  $2|w|$ such that $\even(z) = w$.
\end{lemma}
\begin{proof}
  Since $|w|$ is a multiple of $\ell2^{2\ell}$ and $w$ is $\ell$-perfect,
  for each word $u$ of length $\ell$, $\alocc{w}{u}$ is a multiple of
  $2^\ell$.  Consider a factorization of $w = w_1w_2 \cdots w_r$
 such that  for each~$i$, $|w_i|=\ell$.  Thus,~$r=|w|/\ell$.
  Since $w$ is $\ell$-perfect, for any word~$u$ of length~$\ell$, the set
  $\{ i : w_i = u \}$ has cardinality~$r/2^{\ell}$.  Define $z$ of length
  $2|w|$ as $z = z_1z_2 \cdots z_r$ such that for each $i$,
  $|z_i|=2\ell$, $even(z_i)= w_i$ and for all words $u$
  and~$u'$ of length~$\ell$, the set $\{ i : z_i = u' \vee u \}$ has
  cardinality $r/2^{2\ell}$.  This latter condition is achievable because,
  for each word~$u$ of length~$\ell$, the set $\{ i : \even(z_i) = u \}$
  has cardinality $r/2^{\ell}$ which is a multiple of~$2^{\ell}$, the
  number of possible words~$u'$.
\end{proof}

\begin{corollary} \label{cor:samelength}
  Let $w$ be an $\ell$-perfect word for some even integer~$\ell$.  Then there exists 
   an $\ell$-perfect word~$z$ of length~$2|w|$ such that $\even(z)=w$.
\end{corollary}
\begin{proof}
  Since $w$ is $\ell$-perfect, it is also $\ell/2$-perfect.  Furthermore,
  if $u$ and~$v$ are words of length~$\ell/2$ and~$\ell$ respectively then
  $\alocc{w}{u} = 2^{\ell/2+1}\alocc{w}{v}$.  Thus, the hypothesis of
  Lemma~\ref{lem:doublelength} is fulfilled with~$\ell/2$.
\end{proof}

\begin{corollary} \label{cor:neededsequence}
  There exist a sequence $(w_n)_{n\ge1}$ of words and a sequence of
  positive integers $(\ell_n)_{n\ge1}$ such that $|w_n| = 2^n$,
  $\even(w_{n+1}) =w_n$, $w_n$ is $\ell_n$-perfect and $(\ell_n)_{n\ge 1}$
  is non-decreasing and unbounded.  Furthermore, it can be assumed that
  $w_1 = 01$.
\end{corollary}
\begin{proof}
  We start with $w_1 = 01$, $\ell_1 = 1$, $w_2 = 1001$ and $\ell_2 = 1$.
  For each $n\geq 2$, if $\ell_n2^{2\ell_n}$ divides~$|w_n|$, then
  $\ell_{n+1} = 2\ell_n$ and $w_{n+1}$ is obtained by
  Lemma~\ref{lem:doublelength}.  Otherwise, $\ell_{n+1} = \ell_n$ and
  $w_{n+1}$ is obtained by Corollary~\ref{cor:samelength} .  Note that the
  former case happens infinitely often, so $(\ell_n)_{n\geq 1 }$ is
  unbounded.  Also note that each $\ell_n$ is a power of~$2$.
\end{proof}

\begin{lemma}[Theorem 148~\cite{hardy}]\label{lemma:bb}
  Let $A$ be an alphabet of $b$ symbols.  Let $p(k,r,j)$ be the number of
  words of length $k$ with exactly $j$ occurrences of a given word of
  length~$r$, at any position:
  \begin{displaymath}
    p(k,r,j)=\left |\bigcup_{u\in A^r} \{ w\in A^k : \occ{w}{u}=j\} \right|
  \end{displaymath}
  For every integer $r$ greater than or equal to $1$, for every integer $k$
  large enough and for every real number $\varepsilon$ such that
  $6/\floor{k/r}\leq \varepsilon \leq 1/b^r$,
  \begin{displaymath}
    \sum_{i:\ \abs{i - k/b^r} \ge \varepsilon k} \!\!\!\!\!\!\! p(k,r,i) <
     2\ b^{k+2r-2}r \ e^{- b^r \varepsilon^2 k/6r}.
  \end{displaymath}
\end{lemma}

\begin{lemma}[Theorem~4.6~\cite{Bugeaud12}]\label{lemma:limsup}
  Let $A$ be an alphabet.  An infinite word~$x$ is normal if and only if
  there is a positive number~$C$ such that, for every every word $u$,
  \begin{displaymath}
    \limsup_{N \to\infty} \frac{\occ{x[1..N]}{u}}{N} < \frac{C}{|A|^{|u|}},
  \end{displaymath}
\end{lemma}

Finally, the next lemma  is similar to Lemma~\ref{lemma:limsup} but with  aligned occurrences.
\begin{lemma}\label{lemma:new}
  Let $A$ be an alphabet.  An infinite word~$x$ is normal if and only if
  there is a positive number~$C$ such that, such that for infinitely many
  lengths~$\ell$, for every word~$w$ of length~$\ell$,
  \begin{displaymath}
    \limsup_{N \to \infty} \frac{\alocc{x[1..\ell N]}{w}}{N} < \frac{C}{|A|^\ell}.
  \end{displaymath}
\end{lemma}
\begin{proof}
  The implication from left to right is immediate from the definition of
  normality.  We prove the other.  Fix alphabet $A$ with $b$ symbols, fix
  $x$ and $C$.  Assume  that for infinitely many lengths~$\ell$, for  every word $w$ of length~$\ell$,
%  \begin{displaymath}
%    \limsup_{N \to \infty} \frac{\alocc{x[1..\ell N]}{u}}{N} < \frac{C}{|A|^\ell}.
%  \end{displaymath}
 the stated condition holds.  Equivalently,
  \begin{equation} \label{eq:hyp}\tag{*}
\limsup_{N \to \infty } \frac{\alocc{x [1..N]  }{w}}{N} < \frac{C}{\ell |A|^{\ell}}.
  \end{equation}
We will prove that for every word $u$, of any  length,
  \begin{displaymath}
    \limsup_{N \to \infty } \frac{\occ{x [1..N]  }{u}}{N} < \frac{C}{|A|^{u}}.
  \end{displaymath}
  and  conclude  that $x$ is normal by   Lemma ~\ref{lemma:limsup}.  
The task  is to switch from  aligned  occurrences  to non-aligned occurrences.
For this we  consider  long consecutive  blocks and the fact that  most of them have the expected
 number of non-aligned occurrences of small  blocks inside.

 Fix a length $r$ and a word $u$ of length $r$. 
 Let $\ell$ be any length greater than $r$ for which  (\ref{eq:hyp}) holds.
Fix $\varepsilon$. 
We group the words of length~$\ell$ in  good and bad for $r$ and $\varepsilon$. 
The bad ones deviate from the expected number of occurrences of some word of length~$r$ by  $\varepsilon \ell$ or more.
The good ones  do not.   Lemma~\ref{lemma:bb} bounds the number of these bad words.

We  use  that each bad word has at most $\ell -r+1$ occurrences of~$u$;
each good word  has at most    $\ell/b^r + \varepsilon \ell/b^r$ occurrences of $u$;
and in between any of two consecutive blocks of length  $\ell$  there are at most $r-1$ occurrences of~$u$.

\begin{align*}
\hspace{-1cm}\frac{\occ{x[1..N]}{u}}{N} 
&<\frac{1}{N} \sum_{w\in A^\ell}  \alocc{x[1..N]}{w} (  \occ{w}{u} + (r-1))  
\\
&=\frac{1}{N}  \sum_{\text{bad } w} \alocc{x[1..N]}{w} \big(  \occ{w}{u} + (r-1)\big)  +
     \frac{1}{N} \sum_{\text{good } w} \alocc{x[1..N]}{w} \big(  \occ{w}{u} + (r-1)\big) 
\\
&<\frac{1}{N}   (\ell -r+1 + r-1) \sum_{\text{bad } w} \alocc{x[1..N]}{w}  +
 \frac{1}{N} \left(\frac{\ell}{b^r} + \varepsilon\ell +r-1\right)    \sum_{\text{good } w} \alocc{x[1..N]}{w}  
\\
&<\frac{1}{N} \ \ell \ \big( 2 b^\ell  \  b^{2r-2} r \ e^{-b^r \varepsilon^2 \ell/(6r)}  \big)   \frac{C N}{\ell  b^{\ell}}  +
\frac{1}{N}    \left(\frac{\ell}{b^r} + \varepsilon\ell +r-1\right)        \frac{ N}{\ell} 
\\
&=2   b^{2r-2 } r  e^{-b^r \varepsilon^2 \ell/(6r)}  C +
 \left(\frac{1}{b^{r}} + \varepsilon +\frac{r-1}{\ell}\right).
\end{align*}
For $\ell$ large enough and $\varepsilon=\ell^{-1/3}$ the values
$2 b^{2r-2 } r e^{-b^r \varepsilon^2 \ell/(6r)} C$ and
$ \left( \varepsilon +\frac{r-1}{\ell}\right)$ are arbitrarily small. So,
  \begin{displaymath}
    \limsup_{N\to\infty} \frac{\occ{x[1..N]}{u}}{N} < \frac{C}{b^r}.\qedhere
  \end{displaymath}
\end{proof}

\begin{proof}[Proof of Theorem~\ref{thm:n=2n}]
  Let $(w_n)_{n\ge1}$ be a sequence given by
  Corollary~\ref{cor:neededsequence}.  
  Let $x = 1 1 w_1w_2w_3 \cdots$
  We first prove that $x$ satisfies $x = \even(x)$.  Note that
  $x[2^k+1..2^{k+1}] = w_k$ for each $k \ge 1$ and
  $x[1..2^{k+1}] = 11w_1 \cdots w_k$.  The fact that $w_n = \even(w_{n+1})$
  implies $x[2n] = x[n]$, for every $n \ge 3$.  The cases  for
  $n = 1$ and $n = 2$ hold because  $x[1..4] = 1101$.

  We prove that $x$ is normal.  Consider an arbitrary index
  $n_0$.  By construction, $w_{n_0}$ is $\ell_{n_0}$-perfect and for each
  $n\geq n_0$, $w_n$ is also $\ell_{n_0}$-perfect.  For every word $u$ of
  length $\ell_{n_0}$ and for every~$n\geq n_0$,
   \begin{displaymath}
    \alocc{ x[1..2^{n+1}]}{u} \leq  \alocc{ x[1..2^{n_0}]}{u} + \alocc{w_{n_0} \ldots w_n}{u} 
  \end{displaymath}
  Then, for every $N$ such that $2^{n}\leq  N < 2^{n+1}$  and $n\geq n_0$,
  \begin{align*}
  \frac{\alocc{x[1..N]}{u}}{N/\ell_{n_0}}  
    & \le \frac{\alocc{x[1..2^{n+1}]}{u}}{N /\ell_{n_0}}   \\
    & \le \frac{\alocc{x[1..2^{n_0}]}{u}  +\alocc{w_{n_0} \ldots w_n}{u}}
               {N/ \ell_{n_0}}  \\
    & \le \frac{\alocc{x[1..2^{n_0}]}{u}}{2^n/ \ell_{n_0}} + 
          \frac{\alocc{w_{n_0}\ldots  w_n}{u}}{2^n/\ell_{n_0}}  \\
    & = \frac{\alocc{x[1..2^{n_0}]}{u}}{2^n/ \ell_{n_0}} + 
        \frac{(2^{n_0} + \ldots + 2^n)/(\ell_{n_0} 2^{\ell_{n_0}})}
             {2^n/ \ell_{n_0}}  \\
    & = \frac{\alocc{x[1..2^{n_0}]}{u}}{2^n/ \ell_{n_0}} + 
        \frac{2^{n+1}-2^{n_0}}{2^n 2^{\ell_{n_0}}}   \\
    & < \frac{\alocc{x[1..2^{n_0}]}{u}}{2^n/ \ell_{n_0}} + 
        \frac{2}{2^{\ell_{n_0}}}.
  \end{align*}
  For large values of $N$ and $n$ such that $2^{n}\leq N< 2^{n+1}$, the
  expression 
\[
\frac{\alocc{ x[1..2^{n_0}]}{u} }{ 2^n/ \ell_{n_0}} 
\]
 becomes
  arbitrarily small.  We obtain for every word $u$ of length $\ell_{n_0}$,
  \begin{displaymath}
  \limsup_{N\to \infty} \frac{\alocc{x[1..N]}{u}}{N/ \ell_{n_0}} \leq 3\ 2^{-\ell_{n_0}}.
  \end{displaymath}
  Since the choice of ${\ell_{n_0}}$ was arbitrary, the above inequality holds for
  each $\ell_n$.  
 Since $(\ell_n)_{n\geq 1}$ is unbounded, the hypothesis of  Lemma
 ~\ref{lemma:new} is fulfilled, with $C=3$, so we conclude that $x$ is normal.
\end{proof}

Alexander Shen (personal communication, August 2016) proved that   
almost all binary words satisfying $x = \even(x)$ are normal. 
The  argument in his proof works if the distances between different occurrences of 
 the same repeated symbol grow sufficiently  fast.
For example, his argument can be also used to prove that almost all binary words satisfying 
$x_{3n}=x_n$ are normal.

\section{Almost all pairs are independent}

The next theorem establishes that almost all pairs of infinite words are independent.

\begin{theorem} \label{thm:measureI} 
  The set $I = \{(x,y) : \text{$x$ and $y$ are independent}\}$ has  measure~$1$.
\end{theorem}

To prove it we use that if the oracle~$y$ is
normal then the number of symbols read from~$y$, is linearly bounded by
the number of symbols read from the input~$x$.  This property, stated in
Lemma~\ref{lemma:lemmaK} below, requires the notion of a finite run  and the notion of a forward pair.

A finite run  of a $k$-automaton $\mathcal T$ is a finite sequence of consecutive transitions
\begin{displaymath}
  q_0 \trans{a_{1,1},\ldots,a_{k,1}} q_1 
      \trans{a_{1,2},\ldots,a_{k,2}} q_2 \trans{}\cdots\trans{} q_{n-1} 
      \trans{a_{1,n},\ldots,a_{k,n}} q_n.
\end{displaymath}
The label of the run is the component-wise concatenation of the labels of the transitions.  
More precisely, it is the tuple $\tuple{u_1,\ldots,u_k}$ where each 
$u_j$ for $1 \le j \le k$ is equal to $a_{j,1}a_{j,2}\cdots a_{j,n}$.  
Such a run is written shortly as $q_0 \trans{u_1,\ldots,u_k} q_n$.

Let $\mathcal{T}$ be a $2$-deterministic $3$-automaton whose state
set is $Q$. 
 Let $v$ be a finite word. A~pair $(p,a) \in Q \times A$ is
called a \emph{forward pair} of~$v$ if there is a finite run
$p \trans{au,v|w} q$ for some finite words $u$ and~$w$ and some state~$q$.
A finite word~$v$ is called a \emph{forward word} if it has a maximum
number of forward pairs.
Since the number of pairs $(p,a) \in Q \times A$ is finite, there exist forward words.  
Note that, in case the automaton is total,   every extension of a forward word is also a forward word.
However, not every prefix of a forward word is a forward word.

\begin{lemma}\label{lemma:relation}
  Let $\mathcal{T}$ be a $2$-deterministic $3$-automaton and let $x$ and
  $y$ be two infinite words such that the run
  $\gamma =  q_0 \trans{x,y|z} \limrun$ is accepting. 
 Let $v$ be a forward
  word for~$\mathcal{T}$. For each factorization
  $\gamma =  q_0 \trans{u_1,v_1|w_1} q_1 \trans{u_2,v_2|w_2} q_2
  \trans{x_1,y_1|z_1} \limrun$  such that $v$ occurs in~$v_2$, $u_2$ is
  non-empty.
\end{lemma}

\begin{proof}
  It suffices to prove the result when the word~$v_2$ is equal to~$v$.
  Suppose by contradiction that $u_2$ is empty.  Let $a$ be the first
  symbol of~$x_1$.  Since $u_2$ is empty, the pair $(q_1,a)$ is not
  forward pair of~$v$.  Since $x_1$ is infinite, there exists a
  right extension~$vv'$ of~$v$ such that $(q_1,a)$ is a forward pair
  of~$vv'$.  This contradicts the fact that $v$ has a maximum number
  of forward pairs.
\end{proof}

\begin{lemma} \label{lemma:lemmaK}
  Let $\mathcal{T}$ be a $2$-deterministic $3$-automaton and let $x$ and
  $y$ be two infinite words such that the run
  $\gamma =  q_0 \trans{x,y|z} \limrun$ is accepting.  If $y$ is normal,
  there is a constant~$K$ depending only on~$\mathcal{T}$ such that for
  any factorization $\gamma =  q_0 \trans{u_1,v_1|w_1} q_1 
  \trans{x_1,y_1|z_1} \limrun$
such that  $|u_1|$ is long enough, $|v_1| \le K|u_1|$.
\end{lemma}
\begin{proof}
  Let $v$ be a forward word for $\mathcal{T}$.  Since $y$ is normal there
  is a positive constant~$k$ less than $1$  such that the number of disjoint occurrences of~$v$ in
  $y[1..n]$ is greater than $kn$ for $n$ large enough.  By the previous
  lemma, $|u_1| \ge k|v_1|$ holds.  The result holds then with $K = 1/k$.
\end{proof}

We write $\mu$ for the Lebesgue measure.

\begin{theorem} \label{thm:measure0} 
  For each normal word~$y$, the set $\{ x : \rho(x/y) < \rho(x) \}$ has
  Lebesgue measure~$0$.
\end{theorem}

\begin{proof}
  Fix $y$ normal. Since for every $x$, $\rho(x) \le 1$  (see comment after Definition \ref{def:rho}), it suffices to
  prove $\{ x : \rho(x/y) < 1 \}$ has measure~$0$.  The inequality
  $\rho(x/y) < 1$ holds if there exists a $2$-deterministic $3$-automaton
  that compresses $x$ using $y$ as oracle.  For a $2$-deterministic
  $3$-automaton~$\mathcal{T}$, let $Q$ be its state set and let $K$ be the
  constant obtained by Lemma~\ref{lemma:lemmaK}.  For any integer~$n$ and
  any positive real number $\varepsilon < 1$, let $X_{n,\varepsilon}$ be
  defined by
  \begin{displaymath}
    X_{n,\varepsilon} = \{ x : q_0\trans{u,v|w} q, u = x[1..n], v = y[1..n']
                                \text{ and } |w| < (1-\varepsilon)n \}.
  \end{displaymath}  
  We claim that $\measure(X_{n,\varepsilon}) \le |Q|Kn2^{-n\varepsilon}$.
  The number of configurations $\tuple{q,n,n'}$ is at most
  $|Q|Kn2^{(1-\varepsilon)n}$ because there are $|Q|$ possibles states,
  $n'\leq Kn$, and there are at most $2^{(1-\varepsilon)n}$ words of length
  smaller than $(1-\varepsilon)n$.  Two words $x$ and $x'$ with different
  prefixes of length~$n$ can not reach the same configuration. If they did,
  the rest of the run would be the same and this would contradict the
  injectivity of~$\mathcal{T}$.  Therefore, $X_{n,\varepsilon}$ is
  contained in at most $|Q|Kn2^{(1-\varepsilon)n}$ cylinders of measure~$2^{-n}$.

  Observe that for $n$ fixed, the inclusion
  $X_{n,\varepsilon} \subseteq X_{n,\varepsilon'}$ holds for
  $\varepsilon' \leq \varepsilon$.  Let $\varepsilon(n)$ be a decreasing
  function which maps each integer~$n$ to a real number such that
\[\sum_{n\ge0}{n2^{-\varepsilon(n)n}} < \infty,
\] 
  for instance  $\varepsilon(n) = 1/\sqrt n$.  For each $k$, define
  \begin{displaymath}
    \tilde X_k = \bigcup_{n \ge k} {X_{n,\varepsilon(n)}}.
  \end{displaymath}
  By the choice of the function $\varepsilon(n)$,
  $\lim_{k \to \infty}{\measure(\tilde X_k)} = 0$.  We claim that each
  word~$x$ compressible with $\mathcal{T}$ with normal oracle $y$ belongs
  to $\tilde X_k$ for each $k$.  It suffices to prove that it belongs to
  infinitely many sets $X_{n,\varepsilon(n)}$.  Suppose $x$ is compressed
  with ratio $1-\delta$, for some $\delta>0$.  Then there is an infinite
  sequence of integers $(n_j)_{j\ge0}$ such that the configurations
  $\tuple{q,n_j,n'_j}$ reached after reading the prefix of length~$n_j$ and
  outputting $w_j$, with $|w_j| < (1-\delta)n_j$.  Let $j_0$ be such that
  $\varepsilon(n_{j_0})<\delta$.  Then, for each for $j\geq j_0$, we have
  $x\in X_{n_j,\varepsilon(n_j)}$.  Since this holds for each of the
  countably many $3$-automata~$\mathcal{T}$, we conclude that the measure
  of the set of words compressible with normal oracle $y$ is null.
\end{proof}

\begin{proof}[Proof of Theorem~\ref{thm:measureI}]
  By definition of independence, the complement
  $\bar{I} = A^\omega \times A^\omega \setminus I$ of~$I$ can be decomposed
  as $\bar{I} = J_1 \cup J_2 \cup J_3 \cup J_4$ where the sets $J_1$,
  $J_2$, $J_3$ and $J_4$ are defined by the following equations.
  \begin{alignat*}{2}
    J_1 & = \{ (x,y) : \rho(x) = 0 \} 
        & \qquad\qquad
    J_2 & = \{ (x,y) : \rho(y) = 0 \} \\
    J_3 & = \{ (x,y) : \rho(x/y) < \rho(x) \} 
        & \qquad\qquad
    J_4 & = \{ (x,y) : \rho(y/x) < \rho(y) \} 
  \end{alignat*}
  The sets $J_1$ and~$J_2$ satisfy $\mu(J_1) = \mu(J_2) = 0$.  By symmetry,
  the sets $J_3$ and~$J_4$ satisfy $\mu(J_3) = \mu(J_4)$.  To show 
  that $\mu(I) = 1$, it suffices then to show that $\mu(J_3) = 0$.

  The measure of~$J_3$ is then given by
  \begin{displaymath}
    \mu(J_3) = \int\!\!\!\int_{x,y} 1_{J_3}\, dxdy = \int_y \left(\int_x
      1_{J_3}\,dx \right) dy = \int_y f(y)\,dy
  \end{displaymath}
  where the function~$f$ is defined by $f(y) = \mu(\{x: \rho(x/y) < \rho(x) \})$.
  By Theorem~\ref{thm:measure0}, $f(y)$ is equal to~$0$ for each normal
  word~$y$.  Since the set of normal words has measure~$1$, $f(y)$ is equal
  to~$0$ for almost all~$y$.  It follows that $\mu(J_3) = 0$.  This
  concludes the proof of the theorem.
\end{proof}
\bigskip
\bigskip
\bigskip

\noindent
\textbf{Acknowledgements.}
The authors acknowledge Alexander Shen for many fruitful discussions.
The authors are members of the Laboratoire International
Associ\'e INFINIS, CONICET/Universidad de Buenos Aires--CNRS/Universit\'e
Paris Diderot.  Becher is supported by the University of Buenos Aires and
CONICET.

\bibliographystyle{plain}
\bibliography{independence}
\bigskip

{\setstretch{0.9}
\begin{minipage}{\textwidth}
\noindent
Ver\'onica Becher
\\
Departamento de  Computaci\'on, Facultad de Ciencias Exactas y Naturales \& ICC\\
 Universidad de Buenos Aires\&  CONICET\\
 Argentina.
\\
vbecher@dc.uba.ar
\medskip\\
Olivier Carton
\\
Institut de Recherche en Informatique Fondamentale,
Universit\'e Paris Diderot
\\
Olivier.Carton@irif.fr
\medskip\\
Pablo Ariel Heiber
\\
 Departamento de  Computaci\'on,   Facultad de Ciencias Exactas y Naturales
\\
Universidad de Buenos Aires \& CONICET, Argentina.
\\
pheiber@dc.uba.ar
\end{minipage}
}
\end{document}